\documentclass{article}

\usepackage[a4paper,margin=1.25in]{geometry}
\usepackage{amsmath,amssymb,amstext,amsthm,hyperref,mathtools,natbib}

\usepackage[]{parskip}

\usepackage{bm}
\usepackage{mathtools}
\usepackage{enumerate}
\usepackage{caption}
\usepackage{subcaption}

\usepackage[section]{placeins}

\newtheoremstyle{mystyle} % name
{16pt} % Space above
{10pt} % Space below
{} % Body font
{} % Indent amount
{\bfseries} % Theorem head font
{} % Punctuation after theorem head
{.7em} % Space after theorem head
{\thmname{#1}\thmnumber{ #2}:}% Manually specify head
\newtheoremstyle{withname} % name
{16pt} % Space above
{10pt} % Space below
{} % Body font
{} % Indent amount
{\bfseries} % Theorem head font
{} % Punctuation after theorem head
{.7em} % Space after theorem head
{\thmname{#1}\thmnumber{ #2} (\hspace{-3pt}\thmnote{ #3}):}% Manually specify head

\theoremstyle{mystyle}
\newtheorem{thm}{Theorem}[section]

\theoremstyle{withname}

\theoremstyle{mystyle}
\newtheorem{dfn}[thm]{Definition}

\theoremstyle{mystyle}

\theoremstyle{withname}

\theoremstyle{mystyle}
\newtheorem{cor}[thm]{Corollary}

\theoremstyle{mystyle}
\newtheorem{prop}[thm]{Proposition}

\theoremstyle{withname}

\theoremstyle{mystyle}

\theoremstyle{mystyle}

\theoremstyle{mystyle}
\newtheorem{lemma}[thm]{Lemma}

\newcommand{\re}{\mathbb{R}}

\newcommand{\cs}{\mathcal{S}}

\newcommand{\cum}[2]{C\left({#1} \; ;{#2} \right)}

\newcommand{\leb}{\lambda^{Leb}}

\title{A latent trawl process model for extreme values}
\author{\begin{tabular}{lll}
Ragnhild C. Noven & Almut E. D. Veraart & Axel Gandy\\[10pt]
\multicolumn{3}{c}{Department of Mathematics, Imperial College London} 
\end{tabular}}

\begin{document}
\maketitle

\begin{abstract}
This paper presents a new model for characterising temporal dependence in exceedances above a threshold. The model is based on the class of trawl processes, which are stationary, infinitely divisible stochastic processes.  The model for extreme values is constructed by embedding a trawl process in a hierarchical framework, which ensures that the marginal distribution is generalised Pareto, as expected from classical extreme value theory. We also consider a modified version of this model that works with a wider class of generalised Pareto distributions, and has the advantage of separating marginal and temporal dependence properties. The model is illustrated by applications to environmental time series, and it is shown that the model offers considerable flexibility in capturing the dependence structure of extreme value data.  
\end{abstract}
\section{Introduction}
Modelling dependencies in extreme value data is a topic of growing importance, with applications in a number of fields such as hydrology \citep{DeHaan1998b}, oceanography \citep{Coles1991}, financial risk management \citep{Ledford2003,EMK1997} and environmental science \citep{Davison2012,Heffernan2004}.

The main contribution of this paper is a new extreme value model that can account for serial dependence in the extremes, which extends the hierarchical setup of \cite{Bortot2013}.
The starting point for the model is the observation that the marginal distribution of exceedances converges to the generalised Pareto distribution (GPD), see \cite{Davison1990}. 
In \cite{Bortot2013}, a decomposition of the GPD is used to construct a hierarchical model for exceedances that preserves this distribution marginally. We adopt this hierarchical structure, and then proceed to introduce a new model by using  properties of so-called \emph{trawl processes}.
The original approach of \cite{Bortot2013} involved using a Markov chain to generate dependence in the exceedances; by using a trawl process instead of a Markov chain we obtain a more flexible dependence structure. Moreover, the trawl process framework provides a unified procedure for generating processes with a given infinitely divisible distribution and autocovariance function, whereas \cite{Bortot2013} consider two different specifications for the latent Markov chain with the same marginal distribution and autocovariance function.  

We also consider a modification of the new model that enables it to be used for any generalised Pareto distribution, thus removing the restriction that was inherited from the original model of \cite{Bortot2013} that  the shape parameter of the GPD distribution has to be positive. This means that our modified model can be used for processes with less heavy-tails, which are often found in environmental applications, as illustrated by an application to air pollution data. The modification also improves the interpretability of the parameters, and appears to make the estimation procedure more efficient.

We remark that our new dynamic model for environmental variables such as precipitation and ozone levels can also be used as a basis for designing suitable hedges in terms of weather derivatives. While weather derivatives on rainfall have been traded in the past, there is also recent interest in setting up derivatives to hedge climate variables.

This paper is structured as follows. 
Section \ref{latent_model} introduces the latent trawl process model for dependent extremes in a hierarchical set-up. Section \ref{application} discusses parameter estimation and inference and  develops a measure for the extremal dependence structure that is adapted to the model.
Section  \ref{empirics} shows examples of applying the model to two different environmental time series (rainfall and air pollution), and Section \ref{conclusion} concludes. Some details for the estimation procedure and the proofs of our theoretical results  are presented in Appendix \ref{ltrawlappendix}.
%%%%%%%%%%%%%%%%%%%%%%%%%%%%%%%%%%%%%%%%%%%%

\section{Latent trawl process model} \label{latent_model}
\subsection{Basic structure} \label{setup}
This subsection introduces the hierarchical structure used for the extreme value model, which is taken from \cite{Bortot2013}.  Throughout the article we work on a probability space $(\Omega,\mathcal{F},P)$.
Consider a discrete-time stochastic process denoted by  $\{Y_{j}\}$, which is  assumed to be strictly stationary. We assume that we observe the process at times $j=1,\dots,k$ for $k\in \mathbb{N}$,  and we are interested in extreme values of $Y_{j}$, meaning that $Y_j>u$ for a fixed threshold $u$.  

In order to focus on the extreme values, we will only consider the values and occurrence times of exceedances. To this end, define the exceedances $X_{j}$ by 
\begin{align} \label{exceedances}
X_j:=\max(Y_{j}-u,0), \quad j=1,\ldots,k.
\end{align}
 
From standard extreme value theory (see e.g.~\cite{Pickands1975,Davison1990}), assuming $\{Y_j\}$ are in the domain of attraction of some extreme value distribution, the conditional exceedances $\{X_j|X_j>0\}$ converge to a generalised Pareto distribution (GPD) for an appropriate sequence of thresholds $u_n \rightarrow \infty$. Based on this result we will assume that the conditional distribution of $X_j$ given $X_j>0$ can be approximated by a GPD, for a sufficiently large threshold $u$. The density of the GPD is written as 
\[f_{\text{GPD}}(x|\alpha,\beta)=\frac{\alpha}{\beta}\left(1+\frac{x}{\beta} \right)_+^{-(\alpha+1)}, \quad x \geq 0, \; \alpha,\beta>0,\]
where $y_+=\max(0,y)$, which is a reparametrisation of the standard density with shape parameter $\xi=1/\alpha$ and scale parameter $\sigma=\beta/\alpha$. 

Following \cite{ReissThomas2007} (see also \cite{Bortot2013}), the GPD can be represented as a mixture of an exponential random variable with Gamma distributed parameter, motivating a hierarchical specification for the exceedance process $\{X_j\}$. In particular, we assume that the distribution of   $X_j$ depends on the value of a latent process $\Lambda$ at time $j$, denoted by $\Lambda_j$. This latent process determines both the probability of observing an exceedance, corresponding to $X_j>0$, and the distribution of the exceedances. 

Specifically, we assume that conditionally on the latent process $\Lambda$ the $X_j$ are independent and 
\begin{align}\label{conddistr}
X_j| (X_j>0, \Lambda_j) \sim \text{Exp}(\Lambda_j).
\end{align}
In order to make the threshold exceedances follow the GPD, we can in principle use any stationary stochastic process $\Lambda$ which has gamma marginal law.  The precise specification of $\Lambda$ will be discussed in Section \ref{latent_trawl}.

Next, following \cite{Bortot2013}, we assume that \begin{align}\label{expsurvival}
P(X_j>0|\Lambda_j)=\exp(-\kappa \Lambda_j),
\end{align}
where the parameter $\kappa>0$ is linked to the proportion of the occurrence of exceedances above the threshold $u$. 
Combining \eqref{conddistr} and \eqref{expsurvival}, we can deduce that 
the conditional density of $X_j$ given $\Lambda_j$ is given by 
\begin{align} \label{basicmodel}
\begin{split}
&f(x_j \, |\lambda_j)=\begin{cases}
1-e^{-\kappa \lambda_j}, \quad  &x_j=0,\\[2pt]
e^{-\kappa \lambda_j} \lambda_j e^{-\lambda_j x_j}, \quad & x_j>0,
\end{cases}
\end{split}
\end{align}
where the density $f$ is defined with respect to the measure  $\mu(dx_j)=\delta_0(dx_j)+dx_j$. This construction shows that conditionally on the value of $\Lambda_j$, the exceedance $X_j$ is generated by a two-stage process: at the first stage $X_j$ is set to zero with probability $1-e^{-\kappa \lambda_j}$; at the second stage the distribution of $X_j$ given $\{X_j>0, \, \Lambda_j=\lambda_j\}$ is exponential with parameter $\lambda_j$. The latent process $\Lambda$ may be interpreted as an inverse intensity, as higher values of $\Lambda$ give a lower probability of exceeding the threshold $u$ and a smaller expected value of exceedances. 

Since we required that the observations $X_j$ are independent for distinct $j$, conditional on the corresponding values of $\Lambda$, the conditional joint density of $(X_1,\ldots,X_k)$ factorises and can be written as: 
\begin{align}  \label{independence}
f(x_1,\ldots,x_k\, | \lambda_1,\ldots,\lambda_k)=\prod_{j=1}^k f(x_j|\lambda_j).
\end{align}
This specification implies that any dependence between observations $X_1,\ldots,X_k$ comes from the dependence between corresponding elements of the latent process $\Lambda$. 

To complete the GPD mixture construction, the latent process is required to have a Gamma marginal law, i.e.
%\begin{align*}% \label{gamma}
$\Lambda_j \sim \mbox{Gamma}(\alpha,\beta)$, for $ \alpha,\beta>0$,
which implies that the corresponding density is given by $f_{\Lambda_j}(x)=\beta^{\alpha}\Gamma(\alpha)^{-1}x^{\alpha-1}e^{-\beta x}$ for $x>0$, and the characteristic function is given by 
$\mathbb{E}(\exp(iu \Lambda_j)) = \exp(C(u,\Lambda_j))$,  where  $C(u,\Lambda_j)=-\alpha \log(1-iu/\beta)$ denotes the corresponding cumulant function, which is the  \emph{distinguished logarithm} of the  characteristic function, see \citet[p.~33]{Sato1999}. This specification introduces the restriction $\alpha>0$, which means that the model can only capture data belonging to the Fr\'echet distribution class. Section \ref{transformed} presents a modified version of the model that removes this restriction. 

A straightforward computation shows that, when using the above specification, the exceedances $\{X_j: X_j>0\}$ have a GPD$(\alpha,\beta+\kappa)$ marginal law.
Also. the unconditional probability of observing an exceedance is given by
\begin{align} \label{exprob}
P(X_j>0)=E_{\Lambda}[e^{-\kappa \Lambda}]= \left(1+\frac{\kappa}{\beta} \right)^{-\alpha}.
\end{align}
%%%%%%%%%%%%%%%%%%%%%%%%%%%%%%%%%%%%%%%%%%%%%%%%%%%%%%%%%%%%%%%%%%%%%%%%%%%%%%%%%%%%%%%%%%%%%%%%%%%%%%%%%%%%%%%%%%
\subsection{Latent trawl process} \label{latent_trawl} 
The previous subsection describes a general hierarchical model setup, using the same structure as in \cite{Bortot2013}. So far, the latent process $\Lambda$ has only been specified as having a Gamma marginal law. Now we depart from the approach used in \cite{Bortot2013}, where the latent process is assumed to be a Markov chain (more specifically a  Gaver and Lewis process (G-LP) or a  Warren process (WP)); rather, we consider a new model where $\Lambda$ is a \emph{trawl process}. In principle, any stationary process with Gamma marginal law could be used in this construction, but we will argue in the following that the class of trawl processes is particularly suited   due to the fact that the serial correlation and the marginal distribution can  be modelled independently of each other in the case of trawl processes. 

The conditional independence assumption of the hierarchical model means that any dependence between observations comes from the latent process, and hence this process should have a flexible dependence structure. This motivates the use of a trawl process, which can capture a wide range of dependence structures, as we will discuss in the following. Using a trawl process also means that the observations $X_{j}$ can be seen as coming from a continuous-time process $(X_t)$, which is useful for statistical applications where there may be missing or irregularly spaced data. 

\cite{Bortot2013} consider two particular classes for the latent Markov chain, the Gaver and Lewis process (G-LP) and the Warren process (WP), and proceed to show that these two classes result in different asymptotic properties of the extremes, even though they have the same autocorrelation function. In contrast, the latent trawl process in our model is specified by its trawl set, which corresponds to a particular autocorrelation function. As will be shown in Section \ref{taildep}, the resulting process is asymptotically independent, where the form of the dependence structure is influenced by the trawl set. 

\subsubsection{Definition and properties of trawl processes}
Let us now define the class of trawl processes and present its key properties. 
  Trawl processes have been introduced by \cite{Barndorff-Nielsen2011a} and have been further developped by 
\cite{Barndorff-Nielsen2014, ShephardYang2016a, ShephardYang2016b, Veraart2016}.
 They are stationary infinitely divisible stochastic processes which are made up of two components: the L\'evy basis and the  trawl set. In order to define these components, we need to introduce some notation first. 
 
 To this end, let $S$ be a Borel set in $\re^2$, with the associated Borel $\sigma$-algebra $\mathcal{S}=\mathcal{B}(S)$ and Lebesgue measure $\leb$. Let $\mathcal{B}_b(S)$ be the subsets of $S$ with finite Lebesgue measure, i.e.
$\mathcal{B}_b(S)=\{A \in \mathcal{S}: \leb(A)<\infty\}$.
The purpose of the next definition is to define what we mean by a \textit{homogeneous L\'evy basis}, which is the source of randomness in the trawl process.

\begin{dfn} %\label{rmeasure}
\begin{enumerate}
\item 
A \textit{random measure on} $(S,\mathcal{B}(S))$ is a collection of $\re$-valued random variables 
$\{M(A): A \in \mathcal{B}_b(S) \}$ 
such that for any sequence $A_1,A_2,\ldots$ of disjoint elements of $\mathcal{B}_b(S)$ with  $\cup_{j=1}^\infty A_j \in \mathcal{B}_b(S)$, we have
$M \left(\bigcup_{j=1}^\infty A_j \right)=\sum_{j=1}^\infty M(A_j)$ a.s..
\item
A random measure $M$ on $(S,\cs)$ is \textit{independently scattered} if for any sequence $A_1,A_2,\ldots$ of disjoint elements of $\mathcal{B}_b(S)$, the random variables $M(A_1),M(A_2),\ldots$ are independent. 
\item A random measure $M$ on $(S,\cs)$ is called \textit{infinitely divisible} if for each $n \in \mathbb{N}$ there exist $n$ independent, identically distributed random measures $Z_1^n,\ldots,Z_n^n$ such that
$M \stackrel{d}{=}Z_1^n+\ldots+Z_n^n$.
In particular, infinite divisibility implies that for any finite collection $A_1,\ldots,A_n$ of elements of $\mathcal{B}_b(S)$, the random vector $(M(A_1),\ldots,M(A_n))$ is infinitely divisible in $\re^n$.
\item A random measure on $(S,\mathcal{S})$ is called \textit{stationary} if for any point $s \in S$ and finite collection $A_1,A_2,\ldots,A_n \in \mathcal{B}_b(S)$ such that $A_i + \bm{s} \subset S$, we have that
$(M(A_1+\bm{s}),M(A_2+\bm{s}),\ldots,M(A_n+\bm{s}))\stackrel{d}{=}(M(A_1),M(A_2),\ldots,M(A_n))$.

\item A \textit{homogeneous L\'evy basis} $L$ on $(S,\cs)$ is a random measure that is independently scattered,  infinitely divisible and stationary. 
\end{enumerate}
\end{dfn}

Let $L$ denote a homogeneous L\'evy basis on $(S,\cs)$. Then the characteristic function satisfies the fundamental relation
\begin{align} \label{lseed0}
E \left[\exp\{iuL(A)\} \right]=\exp\left\{\leb(A) K(u)\right\},  \text{ for } A \in \mathcal{B}_b(S),
\end{align} 
where 
\begin{align}\label{lseed1}
K(u) = iu \mu -\frac{1}{2}u^2\sigma^2 + \int_{\re}\left(e^{iuz}-1-iuz\mathbb{I}_{|z|\leq 1}\right)\nu(dz),
\end{align}
for constants $\mu\in \re, \sigma^2\geq 0$ and a L\'{e}vy measure $\nu$, see e.g.~\cite{Rajput1989,Barndorff-Nielsen2011a}. Since equation \eqref{lseed1} has the form of the cumulant function of an infinitely divisible random variable,  we say that we can associate a \emph{L\'{e}vy seed} denoted by $L'$ with the L\'{e}vy basis $L$ which is defined as the random variable whose law is characterised by  \eqref{lseed1}. We then write $C(u,L')=K(u)$ for the corresponding cumulant function and conclude that
\begin{align} \label{lseed}
E \left[\exp\{iuL(A)\} \right]=\exp\left\{\leb(A)C(u,L')\right\},  \text{ for } A \in \mathcal{B}_b(S).
\end{align} 
This shows that  the law of $L(A)$ is fully determined by the L\'evy seed $L'$ and the Lebesgue measure of the set $A$.

We can now define the class of trawl processes. 
\begin{dfn} \label{trawlproc}
Let $A$ be any set in $\mathcal{B}_b(\re \times \re)$, and define a collection of trawl sets $\{A_t\}$ by shifting $A$ along the $\re$-axis corresponding to the last coordinate, which represents time: 
$A_t=A+(0,t):=\{(a_1,a_{2}+t):(a_1,a_2) \in A\}$.
Let $L$ denote a homogeneous L\'evy basis. 
The trawl process
 $(\Lambda_t)_{t\in \mathbb{R}}$ is then defined by evaluating  the homogeneous L\'{e}vy basis over the trawl set, i.e.~by setting $\Lambda_t=L(A_t)$  for $t\in \mathbb{R}$.
% \end{align}
\end{dfn}
%\begin{rmk} \label{ambittrawl}
The trawl process definition can be written as a stochastic integral, which will become useful for calculations in the following. Specifically, we write 
\[\Lambda_t=\int_{\re \times \re} \mathbb{I}_{A_t}(\xi,s) L(d \xi, ds)= \int_{\re \times \re} \mathbb{I}_{A}(\xi,s-t) L(d \xi, ds),\]
where points in $\re^{2}$ are denoted by $(\xi,s)$ for $\xi \in \re, s \in \re$, so the last component corresponds to the time axis. The  stochastic integral is defined in the sense of \cite{Rajput1989}, see also  \cite{Barndorff-Nielsen2012} for a review.

From the definition of the trawl process, we can immediately deduce that the process is stationary and infinitely divisible and that the characteristic function is given by equation \eqref{lseed} since $\Lambda_t \stackrel{\text{d}}{=}\Lambda_0$. Moreover, the  stochastic integral representation  implies that a trawl process is also a so-called mixed moving-average process; it was shown in \cite{Fuchs2013} that such processes are mixing, so it follows that trawl processes are mixing and ergodic. 

{\bf A slice representation for the finite dimensional distributions:}
Next, we study the finite dimensional distributions of a trawl process and derive what we call a \emph{slice representation} for its characteristic function which will be very useful for simulation and inference purposes later on.

To this end, consider a sequence $0\leq t_1\leq \cdots \leq t_k$ with $k\in \mathbb{N}$ and let us now derive the joint characteristic function of $(\Lambda_{t_1},\ldots,\Lambda_{t_k})$.   
We write $\Lambda_j=\Lambda_{t_j}$ to simplify the exposition and note that  typically we will choose $t_j=j$.
We consider the union $A^{\cup,k }:=\cup_{i=1}^k A_{t_i}$.  Using the inclusion-exclusion principle we construct what we call a \emph{slice} partition  
$\{ S_1, \dots, S_{n_k}\}$ of  $A^{\cup,k }$, where $n_k$ denotes the number of elements in the partition: In addition to $\{ S_1, \dots, S_{n_k}\}$ being a partition of  $A^{\cup,k }$, we require that the partition is such that each trawl $A_{t_k}$ can be written as a union of elements of that partition and that the intersection of any number of trawl sets and trawl set complements is a union of subsets in the partition. For general trawls one would need $n_k=2^{k-1}$, whereas for monotonic trawls this number reduces to $k(k+1)/2$.
E.g.~in the case when $k=2$, a suitable slice partition of $A_{t_1}\cup A_{t_2}$ is given by $\{A_{t_1}\cap A_{t_2}, A_{t_1}\setminus A_{t_2}, A_{t_2}\setminus A_{t_1}\}$. 
 \begin{prop}\label{prop:jointchar}
 For $u_1, \dots, u_k \in \re$, we have (using the notation introduced above) that 
\begin{align*}
E\left(\exp\left(i\sum_{j=1}^ku_j\Lambda_j\right)\right)= \exp\left(\sum_{m=1}^{n_k} \leb(S_m) C(\theta_{m}^+ ; L')\right), \quad \text{ for } u_{m}^+:=\sum_{\substack{1\leq j\leq k:\\  A_{t_{j}}\supset S_m}} u_j.
\end{align*}
 \end{prop}

An immediate consequence of Proposition \ref{prop:jointchar} is the following corollary stating the second order properties of a trawl process. 
\begin{cor}
Consider a trawl process with finite second moment. Then for all $t\in \re, h\geq 0$ we have 
$E(\Lambda_t)=\leb(A)E(L')$, $\mathrm{Var}(\Lambda_t) =\leb(A)\mathrm{Var}(L')$, and $\mathrm{Cor}(\Lambda_t,\Lambda_{t+h})=\frac{\leb(A\cap A_h)}{\leb(A)}$.
\end{cor}
\subsubsection{Marginal distribution}

In the context of our latent trawl model, we are exclusively  interested in the case of a marginal Gamma law.
Specifically, we fix a  set $A$ in $\mathcal{B}_b(S)$, and let the L\'evy seed have a normalised Gamma distribution, i.e.~  
\begin{align} \label{gmargins}
L' \sim \text{Gamma}\left(\frac{\alpha}{\leb(A)},\beta \right),
\end{align}
then the trawl process defined by $\Lambda_t=L(A_t)$ has a Gamma$(\alpha,\beta)$ distribution. %; this choice corresponds to setting $\mu=\sigma^2=0$ and $\nu(dz)=\alpha z^{-1}e^{-\beta z}dz$ in \eqref{lseed1}. 

Combining the trawl process $\Lambda$ with the hierarchical model presented in Subsection \ref{setup}, we obtain a stochastic process $(X_j)$ with finite-dimensional densities given by 
\begin{align} \label{latentdens}
&f(x_1,\ldots,x_k)= 
\int_{\re_{+}^k} \left( \prod_{j\in I_0}  (1-e^{-\kappa \lambda_j}) \right)
\left( \prod_{j\in I_>}  \lambda_j e^{-\lambda_j(\kappa+x_j)} \right)
\, dF(\lambda_1,\ldots,\lambda_k), 
\end{align}
where $I_0=\{j\in\{1,\dots,k\}:x_j=0\}$ and $I_> =\{j\in\{1,\dots,k\}:x_j>0\}$.  These densities depend on the joint distribution $F$ of $(\Lambda_1,\ldots,\Lambda_k)$, which is fully specified by the trawl set $A$ and the L\'evy seed $L'$, as shown in Proposition \ref{prop:jointchar}.

\subsubsection{Trawl set} \label{exptrawlset}
To complete the definition of the trawl process $\Lambda$ it now remains to specify the trawl set $A$. For our model we use the so-called exponential trawl set; 
this is the trawl obtained by setting $A=\{(\xi,s): s\leq 0, 0\leq \xi\leq d_{exp}(s)\}\subset [0,1] \times (-\infty,0]$,  for $d_{exp}(s)=\exp(\rho s)$, 
for some $\rho>0$. The resulting process is then called the \textit{exponential trawl process}, and hence we obtain a hierarchical model with parameters $(\rho,\alpha,\beta,\kappa)$.

The autocovariance function of the trawl process is given by $\varphi(h)=\leb(A \cap A_h)$Var$(L')$, and using the exponential trawl gives 
$\leb(A)=\rho^{-1}$, $\leb(A \cap A_h)=e^{-\rho h}
\rho^{-1}$.
Combining with \eqref{gmargins} gives Var$(L')=(\alpha \rho)/\beta^2$, resulting in the autocovariance function
\[\varphi(h)=e^{-\rho h} \frac{\alpha}{\beta^2}=e^{-\rho h} \text{Var}(\Lambda_t).\]
Thus the autocorrelation function of the exponential trawl process has the same decay as the trawl function $d_{exp}$, which is a particular property of the exponential trawl. 

We can also consider a general exponential trawl set, which is constructed from linear combinations of basic exponential trawls. Specifically, define the general exponential trawl set of order $p$ to be bounded above by the function
$d_p(x)=\sum_{i=1}^p w_i e^{\rho_i x}$,
with $\sum_i w_i=1$. The latter restriction is necessary to make the parameters identifiable, as any scaling factor in the weights $w_i$ will scale the area of the trawl, and thus be cancelled by the normalisation in \eqref{gmargins}. When using the general exponential trawl set, the resulting trawl process has an autocorrelation function $r(h)$ given by the corresponding linear combination of $e^{-\rho_i h}$. This construction can be seen as a special case of a superposition-type trawl where the decay parameter $\rho$ is randomised, as in \citet[Section 4]{Barndorff-Nielsen2014}; see also \cite{Barndorff-Nielsen2001a} for a similar approach applied to Ornstein-Uhlenbeck-type processes. Other relevant choices  of the trawl set beyond the exponential setting are discussed in \cite{Barndorff-Nielsen2014} and we also remark that the trawl does not need to be restricted to an $\mathbb{R}^2$-setting, but could also be considered in higher dimensions if necessary.

Summarising the above discussion, we have constructed a trawl process $\Lambda$ with a marginal Gamma distribution and an exponentially decaying autocorrelation function. Using such a discretised trawl process  as the latent process in the hierarchical structure results in a discrete-time process $(X_j)$, where the exceedances $\{X_j>0\}$ have a generalised Pareto distribution. Furthermore, the model allows for dependence between observations $\{X_j\}$, which is derived directly from the dependence in the latent trawl process. 

\subsection{Autocovariance structure} \label{autocovariance}
We now consider the mean and the  autocovariance structure of the exceedance process $(X_j)$ from the latent model, which we summarise in the following proposition.

\begin{prop}\label{acfexceedance}
The mean of the exceedance process $(X_j)$ is for all $j\in \mathbb{N}\cup \{0\}$ given by 
\begin{align*}
E[X]:=E[X_j]&=(1+\kappa/\beta)^{-\alpha}(\beta+\kappa)/(\alpha-1), \quad  \alpha>1.
\end{align*}
As $X$ is a stationary process, it has autocovariance function
$\varphi(h)=E[X_0X_{h}]-E^2[X]$, where for $h\in \mathbb{N}$  
\begin{align}\label{Cov}
E[X_0X_h]
&= \int_\kappa^\infty \int_\kappa^\infty \left(1+\frac{u_0}{\beta} \right)^{b_{0\setminus h}} \left(1+\frac{u_0+u_h}{\beta} \right)^{b_{0,h}} \left(1+\frac{u_h}{\beta} \right)^{b_{h\setminus 0}} \, du_0 \, du_h,
\end{align}
where $b_i=-\alpha \leb(B_i)/\leb(A)$ for $i\in \{(0\setminus h), (0,h), (h\setminus 0) \}$ with  $B_{0 \setminus h}=A_0 \setminus A_h,  B_{0,h}=A_0 \cap A_h,  B_{h \setminus 0}=A_h \setminus A_0$. Note that $b_{0\setminus h}=b_{h\setminus 0}$.
\end{prop}
 
The  integral in \eqref{Cov} can be computed numerically to obtain the autocovariance of $X$ for given parameters $(\alpha,\beta,\rho,\kappa)$, where the parameters $b_i$ are functions of $\rho$ and $h$. 

% The autocovariance function given above for the exceedance process $X$ is discontinuous at zero, which is a consequence of the mixture construction; the latent trawl process itself is continuous in probability. 

The trawl process separates the parameters controlling the marginal and dependence properties of the model. However, this is not the case when considering the full hierarchical model, as the parameters $\alpha,\beta$ and $\kappa$ in the marginal distribution also influence the autocovariance structure of the process. This is illustrated in Figure \ref{autocorr}, which shows two different autocorrelation functions obtained by varying the parameters $\alpha$ and $\beta$ (the plot is provided in a continuous-time setting solely for illustrative purposes). This conflation of marginal and dependence parameters motivates the model in the following subsection. 

\begin{figure}[tbp]
  \centering
  \includegraphics[width=0.6\textwidth]{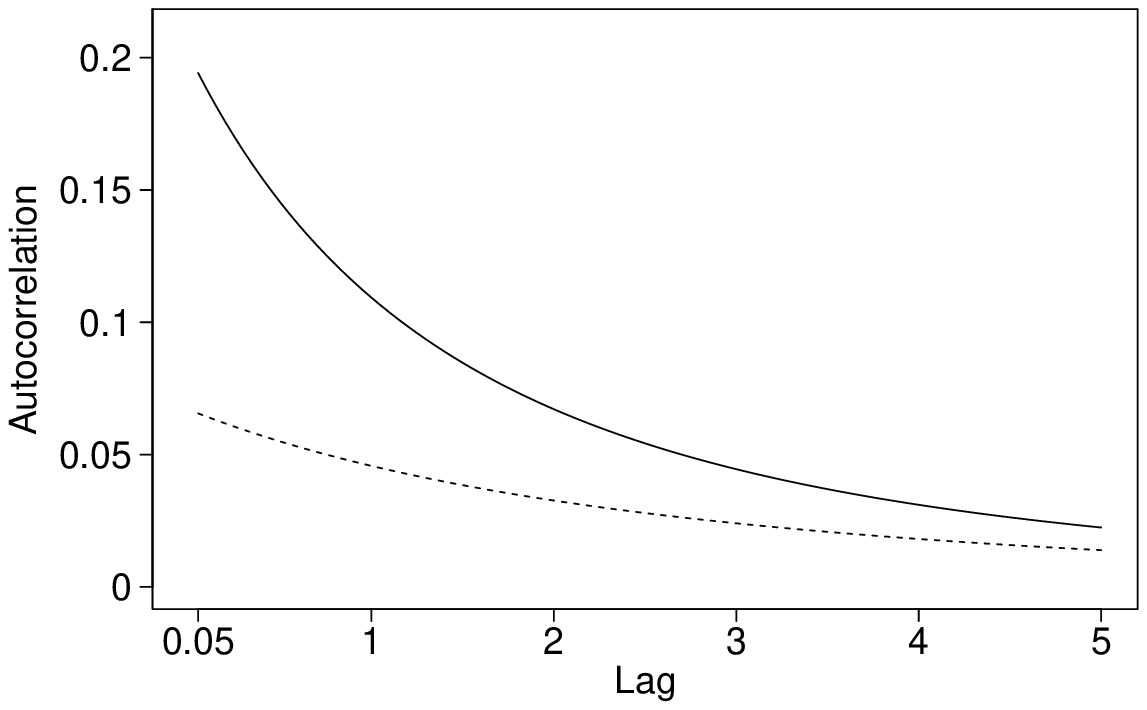}
  \caption[Autocorrelation functions of latent trawl models]{Autocorrelation functions of latent trawl models, the solid line corresponds to a model with $\alpha=4,\beta=4$, the dashed line to $\alpha=9,\beta=1$. Both models set $\rho=0.2$ and $\kappa$ such that the probability of an exceedance equals $0.05$.} 
 \label{autocorr}
\end{figure}   

\subsection{Marginal transformation model} \label{transformed}
This subsection considers a modification of the latent trawl model that has the effect of separating the marginal and dependence properties. This modification allows the model to have generalised Pareto distributions with negative shape parameter. The original restriction to positive values of the shape parameter was highlighted as a potential problem in the conclusion of \cite{Bortot2013}, where a similar modification was suggested but not explored further. 

The model resulting from this modification is also easier to interpret, as the role of each parameter is uniquely defined in terms of controlling either the marginal distribution, probability of exceedance or dependence properties. This also contributes to identifiability of the parameters, in particular we found that the estimation procedure appears to be more efficient for the modified model. 

The modified model is derived from the original model in two steps: the first is to fix the parameters $\alpha,\beta$ of the latent Gamma distribution such that only the parameters associated with the trawl set will influence the trawl process, and thus also the dependence of the exceedances. In the following we will work with $\alpha=\beta=1$, such that the exceedances have marginal law given by GPD$(1,1+\kappa)$. 

The second step is to add an extra layer to the modified model, which uses a standard probability integral transform to give the marginals a GPD$(\xi,\sigma)$ distribution, specifically
\[Z_j=F^{-1}_{\text{GPD}(\xi,\sigma)}(F_{\text{GPD}(1,1+\kappa)}(X_j)):=g(X_j) \quad \xi \in \re, \sigma>0.\]
Note that the above construction implies that  $g(x)=\frac{\sigma}{\xi}\{(1+\frac{x}{1+\kappa})^\xi -1\}$ and $g^{-1}(z)=(1+\kappa)\{(1+z\xi/\sigma )^{1/\xi}-1\}$.
This modified version will be called the marginal transformation (MT) model, it has parameters $(\rho,\kappa,\xi,\sigma)$, where $\xi=1/\alpha, \; \sigma=\beta/\alpha$, and 
conditional density given by 
\begin{align*}
f(z_j|\lambda_j)&=
\left\{ \begin{array}{ll}
1-\exp\{-\kappa\lambda_j\}, & z_j=0, \\[6pt] 
 J(z_j) \, \lambda_j \exp\{-\lambda_j(\kappa+g^{-1}(z_j))\},& z_j>0 ,\\[6pt]
\end{array}  \right.
\text{ where } J(z_j)&=\frac{f_{\text{GPD}(\xi,\sigma)}(z_j)}{f_{\text{GPD}(1,1+\kappa)}(g^{-1}(z_j))},
\end{align*}
with respect to the measure $\mu(dz_j)=\delta_0(dz_j)+dz_j$.  This follows by noting that the transformation $g$ maps the event $\{X=0\}$ to $\{Z=0\}$, thus leaving the atom at zero unchanged, whereas the transformation of the continuous part on $\{X>0\}$ introduces a standard Jacobian term. 

 Now let 
 $I_0=\{j\in\{1,\dots,k\}: z_j=0\}$ and $I_>=\{j\in\{1,\dots,k\}: z_j>0\}$. Then the 
finite-dimensional densities of the MT model can be represented  by
\begin{multline*}
f(z_1,\ldots,z_k)
\\
=
\int_{\re_{+}^k} \left( \prod_{j\in I_0} \, \left(1-\exp\{-\kappa\lambda_j\}\right) \right) \left(\prod_{j\in I_>}  \, J(z_j) \, \lambda_j \exp\{-\lambda_j(\kappa+g^{-1}(z_j))\} \right) \, dF(\lambda_1,\ldots,\lambda_k).
\end{multline*}
When using the MT model in an application, we note that the empirical observations of the exceedances will be described by the $(Z_j)$ and not by the $(X_j)$ as in the earlier model specification.

\section{Model fitting and evaluation} \label{application}

\subsection{Pairwise likelihood} \label{pairlik}
We now consider parameter estimation for the latent trawl model described above. 
The parameter vector of interest is denoted by  $\theta=(\rho,\kappa,\xi,\sigma)^{\top}\in \Theta$, where $\Theta\subset \re^4$ denotes the parameter space.

Given observations $\{Y_j\}$ from a stationary time series, we transform them as in Section \ref{setup} to get exceedances $X_j:=\max(Y_{j}-u,0), \, j=1,\ldots,k$. We also assume there are $l$ positive observations $X_{p_1},\ldots,X_{p_l}$, and $m$ observations $X_{q_1},\ldots,X_{q_m}$ taking the value zero, which we will call exceedances and non-exceedances, respectively. Clearly, $l+m=k$.

The likelihood of the observations $\{X_j\}$ under the original latent trawl model now follows from \eqref{latentdens}, and can be written as 
\begin{align*} 
f(x_1,\ldots,x_k)&=\int_{\re^k_+} \prod_{r=1}^l \lambda_{p_r} \exp\{-(x_{p_r}+\kappa) \lambda_{p_r}\} \prod_{s=1}^m \left(1-\exp\{-\kappa \lambda_{q_s}\}\right) \, dF(\lambda_1,\ldots,\lambda_k), 
\end{align*}
where $F$ is the joint density function of the trawl process observations $\Lambda_1,\ldots,\Lambda_k$. The integrand above can be expanded to obtain a sum involving $2^m$ number of terms, by first defining $\mathcal{S}_t$ as the collection of subsets of $\{q_1,\ldots,q_m\}$ of size $t$, and letting $u_{r}=x_{p_r}+\kappa$, 
to give
\begin{align*}
\sum_{t=0}^m (-1)^t \sum_{\pi_t \in \mathcal{S}_t} \prod_{r=1}^l \lambda_{p_r} \exp\left\{- u_{r}\lambda_{p_r}\right\} \prod_{s_j \in \pi_t} \exp\left\{- \kappa \lambda_{s_j} \right\},
\end{align*}
with the convention $\prod_{s_j \in \pi_0} (\ldots) =1$. This can be rewritten as
\begin{align*}
\sum_{t=0}^m \sum_{\pi_t \in \mathcal{S}_t}(-1)^{t+l} \frac{\partial}{\partial u_{1}}\ldots\frac{\partial}{\partial u_{l}} \exp\left\{-\sum_{r=1}^l u_{r} \lambda_{p_r} -\sum_{s_j \in \pi_t} \kappa \lambda_{s_j}\right\}.
\end{align*}
Using this expression and exchanging integrals and partial derivatives we see that the full likelihood reduces to
\begin{align*}
&\sum_{t=0}^m \sum_{\pi_t \in \mathcal{S}_t}(-1)^{t+l} \, \frac{\partial}{\partial u_{1}}\ldots\frac{\partial}{\partial u_{l}} \int_{\re^k_+}  \; \exp\left\{-\sum_{r=1}^l u_{r} \lambda_{p_r} -\sum_{s_j \in \pi_t} \kappa \lambda_{s_j}\right\}  \, dF(\lambda_1,\ldots,\lambda_k)\\
=&\sum_{t=0}^m \sum_{\pi_t \in \mathcal{S}_t}(-1)^{t+l} \, \frac{\partial}{\partial u_{1}}\ldots\frac{\partial}{\partial u_{l}} \; E\left[ \exp\left\{-\sum_{r=1}^{l} u_{r} \Lambda_{p_r} -\sum_{s_j \in \pi_t} \kappa \Lambda_{s_j}\right\} \right],
\end{align*}
where the expectation is with respect to the corresponding variables $\{\Lambda_{p_r}\}$ and $\{\Lambda_{s_j}\}$ determined by $\pi_t$. 

The expected values in these terms are just joint Laplace transforms, which can be derived from the joint characteristic functions given in Proposition \ref{prop:jointchar}, which involve the parameters of the L\'evy seed $L$ and trawl intersection areas; hence the complete likelihood reduces to a sum of partial derivatives of Laplace transforms. 

The likelihood as given above is not easy to compute in practice, for two reasons. First, it would require multiple numerical partial derivatives to be performed. Second, the number of non-exceedances $m$ is usually close to the number of observations $k$; this is because the latent model is defined to have a GP distribution, and to justify this assumption we need to consider a sufficiently high threshold for exceedances, often given by a large percentile of the observations. Now the two first sums in the likelihood above have $2^m$ terms in total, and hence the likelihood becomes computationally intractable for any reasonable sample size $k$. 

Because of the computational issues associated with the sample size, we consider using a pairwise likelihood approach for model fitting, which is a particular kind of composite likelihood \citep{Varin2008,Cox2004,Varin2011a}. As stated in \cite{Varin2008}, composite likelihood estimators have good properties when the data can be seen as consisting of roughly independent blocks, i.e.~the autocorrelation function decays sufficiently fast. Thus using the pairwise likelihood should provide reasonable performance for the latent trawl model with an exponential trawl set. 
 
Given observations $x_1,\ldots,x_k$, the pairwise likelihood $f_{PL}$ for the parameter vector $\theta=(\rho,\kappa,\xi,\sigma)^{\top}\in \Theta\subset \re^4$ takes the form
\[f^\triangle_{PL}(\theta|x_1,\ldots,x_k)=\prod_{i=1}^{k-1} \prod_{j=i+1}^{\min(i+\triangle,k)} f(x_i,x_j),\]
where $f(\cdot)$ is the original bivariate density function and $\triangle$ denotes the maximum separation between observations. For the latent trawl model, each pairwise likelihood term $f(x_i,x_j)$ involves at most 4 terms in the sum above, and so these terms can be evaluated explicitly in terms of the parameters of the trawl process. There are four different cases, as each of $x_i$ and $x_j$ can be an exceedance or not; the explicit forms of $f(x_i,x_j)$ are given in Appendix \ref{AppendixLike}. 
We denote by 
\begin{align*}
\widehat \theta = \arg \max_{\theta}f^\triangle_{PL}(\theta|x_1,\ldots,x_k)
\end{align*}
the maximum pairwise likelihood estimator.

We note that according to \cite{Cox2004}, the pairwise likelihood estimator is unbiased and asymptotically normal under the usual regularity conditions. 
When looking at the asymptotic theory  in this context, we assume that we have fixed the threshold when computing the relevant exceedances. We are not allowing for a double asymptotic setting  where the threshold is increasing at the same time as the number of observation. A more detailed investigation of such a double asymptotic is beyond the scope of this article. 

\subsection{Conditional tail dependence coefficient} \label{taildep}
We now consider the extremal dependence structure of our model. A common measure of dependence at high levels is the extremal index \citep{Leadbetteretal1983}, which can be characterised as
\[\theta=\lim_{n \rightarrow \infty} P(X_i \leq u_n, 2 \leq i \leq l_n|X_1>u_n),\]
where $u_n$ is an increasing sequence of thresholds and $l_n=o(n)$ (\cite{ancona2000comparison}, following \cite{OBrien}). 

The extremal index $\theta$ essentially describes the dependence across blocks of observations whose length tends to infinity; it can also be defined as the reciprocal mean cluster length, where a cluster is a collection of exceedances in a block. Thus to estimate the extremal index one has to consider very high-dimensional joint distributions, which makes it analytically intractable in many cases, in particular for our latent trawl model. For the applications in the following subsections we will instead consider simulation-based estimates of the extremal index.  

There are also measures of extremal dependence that work on a shorter range of observations than the extremal index: \cite{Coles1999} quantify the dependence between the extreme values of two random variables $X_1,X_2$ in the upper tail-dependence coefficient, given by \footnote{this quantity is sometimes denoted $\lambda_U$, and should not be confused with the ``coefficient of tail dependence'' defined in \cite{Ledford1996}.} 
\[\chi = \lim_{u \rightarrow 1} P(F(X_2)>u | F(X_1)>u),\]
where $F$ is the common marginal distribution of $X_1,X_2$. In other words, $\chi$ gives the limiting probability of  $X_2$ exceeding the threshold $u$ given an exceedance $X_1$, with both variables on a uniform scale. When $X_1,X_2$ come from a stationary time series, $\chi$ can be seen as the probability of observing consecutive exceedances given a single exceedance. To get a broader characterisation of the extremal dependence one can also consider the complete function 
\[\chi(u_1,u_2)=P(F(X_2)>u_2 | F(X_1)>u_1),\]
defined on $[0,1]^2$, where $\chi(u_1,0)=1$ and $\chi(0,u_2)=1-u_2$. 

We would like to use the function $\chi(u_1,u_2)$ and the limiting measure $\chi$ to evaluate the dependence between two observations $X_1,X_2$ from the latent trawl model as a function of the lag $t_2-t_1$. However, the above definition is based on the assumption that $X_1,X_2$ are continuously-valued random variables with distribution function $F$ defined on their entire range, and this does not fit with the exceedance framework where the limiting distribution can only be assumed to hold above a threshold. Some extreme value models  (see e.g.~\citet{Ledford1996}) do not consider exceedances separately, but specify the same probability density function $f$ for the whole range of observations, and then treat observations below the threshold $u$ as censored, so they have probability
$\int_0^u f(s) ds$.  
Our model differs in that it models the occurrence of exceedances explicitly, resulting in a distribution for exceedances $X$ only with an atom at zero, and hence the standard definition of $\chi$ cannot be applied directly. However, we can construct an analogue to the tail dependence coefficient by conditioning on both exceedances being positive, but care must be taken to ensure that the resulting measure is uniform on the $u_2$ margin, as shown in the following definition. 
\begin{dfn}
Consider the latent model as defined in Section \ref{setup}
and set 
\begin{align}\label{F2e_1}
F_{2e}(x)=P(X_0 \leq x\, | X_0>0, X_h>0)=P(X_h \leq x\, | X_0>0, X_h>0), \quad \text{ for } h\in \mathbb{N}.
\end{align}
The \textit{conditional tail dependence function} $\varphi$
is defined as 
\begin{align*}
\varphi(h,u_1,u_2):= P(F_{2e}(X_h)>u_2 \, | F_{2e}(X_0)>u_1, X_0>0,X_h>0), \text{ for } 0\leq u_1, u_2 \leq 1,
\end{align*}
 and the  \textit{conditional tail dependence coefficient} is defined as 
\begin{align*}
\varphi(h):=&\lim_{u\uparrow 1} \varphi(h,u,u).
\end{align*}
\end{dfn}
Note that we show in Lemma \ref{Fe2Lemma} in Appendix \ref{AProofs} that the identity \eqref{F2e_1} holds.

The conditional tail dependence function  can be calculated explicitly in terms of the parameters of the latent trawl model.
 Specifically, we have the following three key results, which are proved in the appendix.
\begin{prop}\label{propcondtail}
Let $h\in \mathbb{N}$ and set $B_{0\setminus h}=A_0 \setminus A_h,  B_{0,h}=A_0 \cap A_h, B_{h \setminus 0}=A_h \setminus A_0$ and $b_i=-\alpha \leb(B_i)/\leb(A)$ for $i\in \{(0\setminus h), (0,h), (h\setminus 0) \}$.
The conditional tail dependence function for $X_0$ and $X_h$ in the latent trawl model is given by 
\begin{align*}
\varphi(h,u_1,u_2)=\left(1+\frac{F_{2e}^{-1}(u_2)}{\beta+2\kappa+F_{2e}^{-1}(u_1)}\right)^{b_{0,h}}\left(1+\frac{F_{2e}^{-1}(u_2)}{\beta+\kappa}\right)^{b_{h\setminus 0}}, \quad \text{ for } 0\leq u_1, u_2 \leq 1.
\end{align*}
\end{prop}

\begin{prop}\label{scalingofindex}
The conditional tail dependence function satisfies the same marginal scaling as the original tail dependence index $\chi$, namely $\varphi(h,u_1,0)=1,\varphi(h,0,u_2)=1-u_2$ for any $0\leq u_1, u_2 \leq 1$.
\end{prop}

\begin{thm}\label{thm_phi}
For the original latent trawl model, we have that 
$\phi(h)=0$ for any $h\in \mathbb{N}$; so according to this measure the model is asymptotically tail independent. 
\end{thm}

Note that the speed of decay of the conditional tail dependence function to zero increases with the value of $b_{h\setminus 0}$; in other words, the larger the intersection of the trawl sets given by $X_0,X_h$, the slower the model decays to independence as the threshold increases. 

%\todo{Explore asymptotics, use $u \sim n^{1/\alpha}$, which is correct normalisation for convergence to Frechet GEV distribution. Comment that we get asymptotic independence}

It has been pointed out by \cite{Coles1999} that the class of asymptotically independent distributions is of fundamental importance in multivariate extreme value theory, see also
\cite{Ledford1996,Ledford1997,BruunTawn1998,Bortot1998}. 
%\newpage
\section{Empirical examples}\label{empirics}
\subsection{Heathrow data}
In this subsection we use the latent trawl model to analyse a data set consisting of daily accumulated rainfall amounts at Heathrow (UK) over the years 1949-2012, provided by the \cite{MIDAS}. We set the threshold $u$ at the 95 percentile of the original data ($8.9$mm), resulting in the time series of exceedance values shown in Figure \ref{hts}. 

We fitted both the latent trawl model described in Section \ref{latent_trawl} and the latent Markov chain model of \cite{Bortot2013}, using both the G-LP and Warren process for the Markov chain. The parameter estimation was done using pairwise likelihood as described in Section \ref{pairlik}, with separation parameter $\triangle=4$ (motivated by our simulation experiments); the resulting estimates are shown in Table \ref{hresults}. We see that the marginal parameters are similar across the models, which is reasonable as the models all have a marginal GPD$(\alpha,\beta+\kappa)$ distribution and $\kappa$ controls the marginal exceedance probability. 

The parameter $\rho$ controls the latent dependence structure of the models. For the latent trawl process it is the decay parameter of the exponential trawl function, whereas for the latent Markov chains it enters in the autocorrelation function $\phi(h)=\rho^h=\exp(h \log(\rho))$. Using this relation we see from the fitted values of $\rho$ that all three latent processes have similar autocorrelation functions, where the G-LP and WP processes have the fastest and slowest decay, respectively. 

\begin{figure}[tbp]
  \centering
  \includegraphics[width=0.55\textwidth]{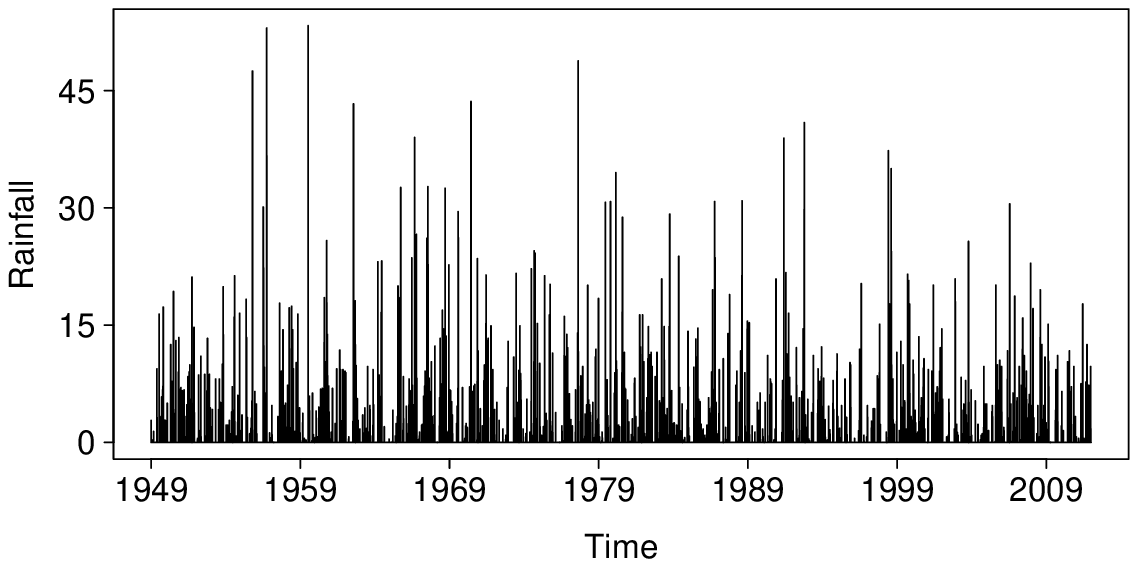}
  \caption{Heathrow rainfall exceedances.}
\label{hts} 
\end{figure}

\begin{table}[tbp]
\centering
\begin{tabular}{rrrrr}
  \hline
 & $\alpha$ & $\beta$ & $\rho$ & $\kappa$ \\ 
  \hline
Latent trawl & 6.33 & 20.12 & 0.27 & 12.18 \\ 
  G-LP & 6.43 & 20.64 & 0.70 & 12.25 \\ 
  WP & 6.30 & 19.94 & 0.78 & 12.15 \\ 
   \hline
\end{tabular}
\caption[Parameter estimates for exceedance models]{Estimated parameters for Heathrow rainfall data using the latent trawl model and the latent Markov chain model with G-LP and WP chains.}
\label{hresults}
\end{table}

Figure \ref{heind} shows estimates of the extremal index for the latent trawl and latent Markov chain models. These estimates are based on simulating time-series of length 1,000,000 from the fitted models, and then estimating $\theta$ as the inverse cluster length using the R package \verb+evd+, where a cluster is defined as ending when three consecutive exceedances fall below the threshold. 

These estimates indicate that the latent trawl and latent Markov chain models both manage to capture the main dependence structure in the extremes. The only visible difference occurs at high thresholds, where the G-LP model appears to underestimate $\theta$, i.e.~overestimate the dependence. This discrepancy at high levels could be explained by the fact that the G-LP model is asymptotically dependent, as shown in \cite{Bortot2013}, whereas the empirical estimates of $\theta$ indicate that the rainfall data is asymptotically independent. 

\begin{figure}[htbp]
  \centering
  \includegraphics[width=0.6\textwidth]{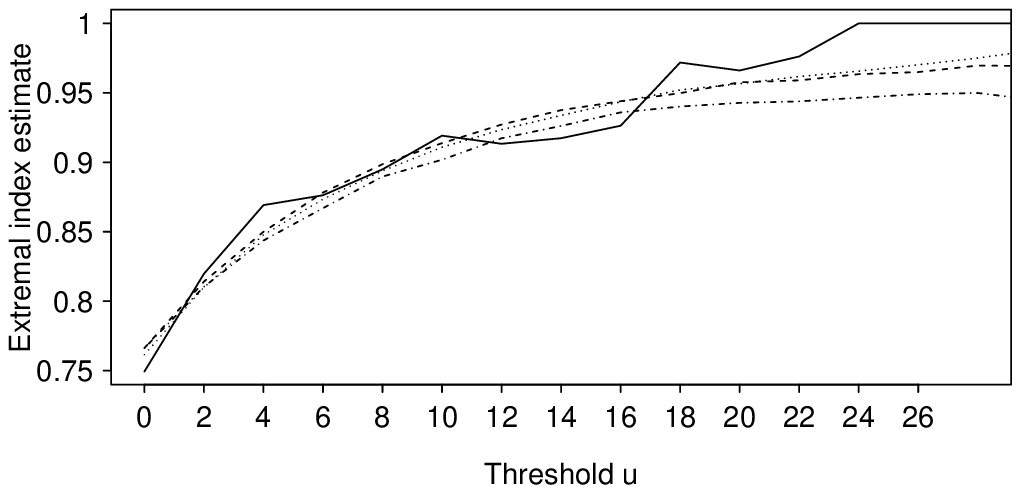}
  \caption[Estimated extremal index]{Estimated extremal index for empirical rainfall data (solid), latent trawl (dashed), G-LP (dash-dotted) and WP (dotted).}
  \label{heind}
\end{figure}

\subsection{Pollution data} \label{pollution_section}
We now consider a second application that illustrates the marginal transformation model introduced in Section \ref{transformed}. We use this model to analyse measurements of Ozone levels in Bloomsbury, London; specifically, the data gives the daily maximum of the 8-hour running mean, measured in units of $\mu g/m^3$ (microgrammes per cubic metre). They are obtained from the \cite{air}.  

Although there is evident seasonality in the original data, the effect diminishes significantly as the threshold increases, and so we do not adjust for seasonality in the extreme values for this application, but refinements of this approach could be considered in future work. We chose the threshold $u$ to be the 97 percentile of the original data ($81 \mu g/m^3$), resulting in the exceedances shown in Figure \ref{ozone}. 

\begin{figure}[tbp]
  \centering
  \includegraphics[width=0.55\textwidth]{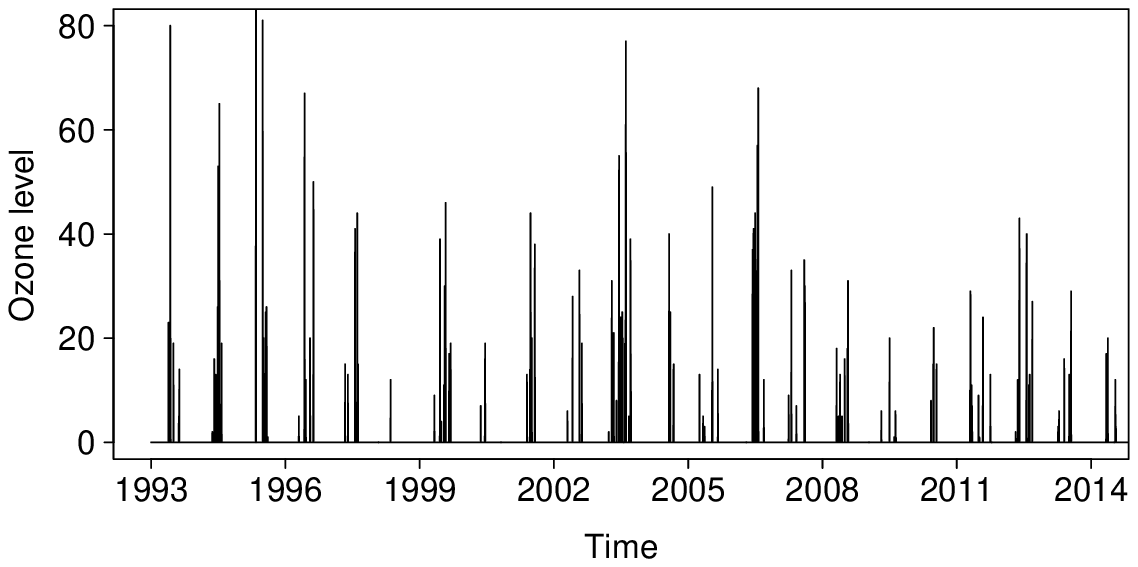}
  \caption{Ozone level exceedances}
  \label{ozone}
\end{figure}
 
Fitting the GPD distribution to the data directly indicates a negative $\xi$-value, i.e.~a finite upper bound, which cannot be captured by the standard model. Thus we use the transformed model described in section \ref{transformed}, which corresponds to taking the basic hierarchical model (either based on latent trawl or latent Markov chain as in \cite{Bortot2013}) and fixing $\alpha, \beta=1$, and then transforming the marginals to GPD$(\xi,\sigma)$.  

The models were fitted using pairwise likelihood with $\triangle=4$. Table \ref{ozest} shows the resulting estimates for the marginal transformed versions of the latent trawl and latent Markov chain models.

Figure \ref{pollution} shows estimates of the extremal index based on simulations of length 1,000,000, obtained by the same method as for the rainfall data above. The estimates of the extremal index indicate that the transformed latent Markov chain model does not accurately capture the extremal dependence structure in this example. In particular, the Warren process model appears to underestimate the extremal dependence (i.e.~$\theta$ is too high), whereas the opposite is the case for the GL-P model. 

\begin{table}[ht]
\centering
\begin{tabular}{rrrrr}
  \hline
 & $\xi$ & $\sigma$ & $\rho$ & $\kappa$ \\ 
  \hline
Latent trawl & -0.11 & 20.73 & 0.17 & 32.69 \\ 
G-LP & -0.04 & 21.09 & 0.56 & 32.19 \\ 
WP & -0.11 & 20.74 & 0.96 & 32.44 \\ 
   \hline
\end{tabular}
\caption[Parameter estimates for transformed exceedance models]{Estimated parameters for Ozone data using transformed versions of latent trawl model and latent Markov chain model with G-LP and WP chains.}
\label{ozest}
\end{table}

These results show that when the marginal parameters $\alpha,\beta$ are fixed in the latent layer, the latent Markov process model has less flexibility in the dependence structure than the latent trawl process model. This indicates that in the original hierarchical structure, the marginal parameters also have a strong influence on the dependence structure, and this contributes to the flexibility of the model. The transformed model has the advantage that its parameters are clearly interpretable as contributing to either the dependence or the marginal distribution. In our experience, the parameter estimation procedure for the transformed model also appears to be more reliable.  

\begin{figure}[tbp]
  \centering
  \includegraphics[width=0.6\textwidth]{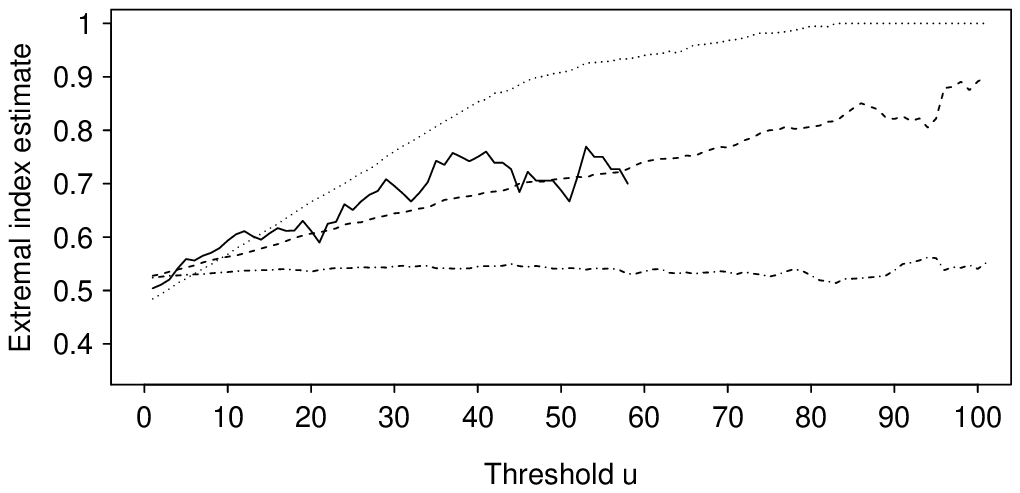}
  \caption[Estimated extremal index for transformed model ]{Estimated extremal index for empirical Ozone data (solid), and transformed versions of the latent trawl (dashed), G-LP (dash-dotted) and WP (dotted) models.}
\label{pollution} 
\end{figure}

\section{Conclusion}\label{conclusion}
In this paper we have investigated a new model for time series of extremes based on trawl processes. 
% The first part of the paper focused on trawl processes in the context of statistical models, showing how they can be simulated and a suitable process can be chosen to obtain the desired autocovariance structure. Based on these results 
We constructed an extreme value model that uses the trawl process framework to obtain a flexible dependence structure. This was done by replacing the latent Markov chain in the setup of \cite{Bortot2013} with a trawl process. In contrast with other hierarchical models, this construction has the advantage of preserving the generalised Pareto distribution for the marginals, which is consistent with extreme value theory.  

We have also considered a modification to the original model structure that extends the parameter space, allowing for negative shape parameters. To evaluate the extremal dependence we have also developed an adapted version of the tail-dependence coefficient, which can be evaluated analytically for the trawl process model. 

The original and modified models were used to analyse two environmental time series, and compared with the latent Markov chain models of %Bortot and Gaetan~
\cite{Bortot2013}. For the application using the original model, the results were very similar in terms of capturing the extremal dependence of the data. The advantage of using the latent trawl process was clearer when using the transformed model, where the trawl-based model performed better due to the added flexibility in the latent dependence structure.  

There are several aspects of the latent trawl model that could benefit from further investigation. For example, we have used a simple exponential trawl throughout this paper; it would be interesting to look at the result of using different parameterisations for the trawl set. Another possibility is to consider a model where the threshold $u$ is allowed to vary, and then letting the trawl set depend on $u$, which should result in a wider range of dependence levels across thresholds. 

%%%%%%%%%%%%%%%%%%%%%%%%%%%%%%%%%%%%%%%%%%%%%%%%%%%%%%%%%%%%%%%%%%%%%%%%%%%%%%%%%%%%%%
%%%%%%%%%%%%%%%%%%%%%%%%%%%%%%%%%%%%%%%%%%%%%%%%%%%%%%%%%%%%%%%%%%%%%%%%%%%%%%%%%%%%%%

%%%%%%%%%%%%%%%%%%%%%%%%%%%%%%%%%%%%%%
\begin{appendix}
\section*{Acknowledgments}
R. C. Noven gratefully acknowledges financial support from the Grantham Institute for Climate Change, Imperial College London. A. E. D. Veraart acknowledges financial support by a Marie Curie FP7 Integration Grant within the 7th European Union Framework Programme. We thank the UK Meteorological Office, the British Atmospheric Data Centre and the London Air Quality Network for providing the data used.

\section{Appendix} \label{ltrawlappendix}
\subsection{Pairwise likelihood in the latent trawl model }\label{AppendixLike} 
The pairwise likelihood can be derived by using Proposition \ref{prop:jointchar} in the likelihood expression. Consider two  observations $x_i,x_j$ corresponding to two distinct time points $t_i=i,t_j=j$ for $i,j \in \{1, \dots, k\}$. Then we can write down the explicit partition of the trawl sets as follows:
$\{B_{i\setminus j}=A_{t_i} \setminus A_{t_j}, B_{i,j}=  A_{t_i} \cap A_{t_j},  B_{j\setminus i}= A_{t_j} \cap  A_{t_i}\}$. Furthermore, we note that $\leb(B_{i\setminus j}) =\leb(B_{j\setminus i})$ and we set $b_{i,j}=-\alpha \leb(B_{i,j})/\leb(A)$ and  $b_{i \setminus j}=-\alpha \leb(B_{i\setminus j})/\leb(A)=-\alpha \leb(B_{j\setminus i})/\leb(A)=b_{j \setminus i}$.
The joint likelihood of $x_i,x_j$ is given below for each of the four possible cases.

First, we consider the case of no exceedances, so $X_i=X_j=0$. Then
\begin{align*}
f(x_i,x_j)&=E_{\Lambda_i,\Lambda_j} \left[ (1-e^{-\kappa \Lambda_i})(1-e^{-\kappa \Lambda_j}) \right] %\\
&=1-2 \left(1+\frac{\kappa}{\beta} \right)^{-\alpha}+\left(1+\frac{\kappa}{\beta}\right)^{2b_{i\setminus j}}\left(1+\frac{2\kappa}{\beta}\right)^{b_{i,j}}.
\end{align*}
Second and third, we focus on the two cases of one exceedance. Here we consider $X_i>0,X_j=0$, the case $X_i=0,X_j>0$ is similar. We have 
\begin{align*}
f(x_i,x_j)&=E_{\Lambda_i,\Lambda_j} \left[ e^{-\kappa \Lambda_i} \Lambda_i e^{-\Lambda_i x_i} \left(1-e^{-\kappa \Lambda_j} \right) \right] \\
&=E_{\Lambda_i,\Lambda_j} \left[ \Lambda_i e^{-(\kappa+x_i)\Lambda_i} \right]-E_{\Lambda_i,\Lambda_j} \left[\Lambda_i e^{-\Lambda_i(\kappa+x_i) -\kappa \Lambda_j} \right] \\
&=-\left. \frac{\partial}{\partial u_i}  E_{\Lambda_i,\Lambda_j} \left[e^{-(u_i\Lambda_i+u_j \Lambda_j)} \right] \right|_{\substack{u_i=\kappa+x_i,\\u_j=0}}
%\\
%&
+\left. \frac{\partial}{\partial u_i}  E_{\Lambda_i,\Lambda_j} \left[e^{-(u_i\Lambda_i+u_j \Lambda_j)} \right] \right|_{\substack{u_i=\kappa+x_i,\\ u_j=\kappa}} %.  
% \end{align*}  
% We have that
% \begin{align*}
% &\frac{\partial}{\partial u_1}  E_{\Lambda_1,\Lambda_2} \left[e^{-(u_1\Lambda_1+u_2 \Lambda_2)} \right]\\
% &=-\frac{\alpha \rho |B_1|}{\beta} \left(1+\frac{u_1}{\beta}\right)^{-\alpha \rho|B_1|-1} \left(1+\frac{u_1+u_2}{\beta} \right)^{-\alpha \rho|B_2|}\left(1+\frac{u_2}{\beta}\right)^{-\alpha \rho|B_3|}\\
% &-\left(1+\frac{u_1}{\beta}\right)^{-\alpha \rho |B_1|} \frac{\alpha \rho |B_2|}{\beta} \left(1+\frac{u_1+u_2}{\beta} \right)^{-\alpha \rho |B_2|-1} \left(1+\frac{u_2}{\beta}\right)^{-\alpha \rho |B_3|}.
% \end{align*}
% Evaluating, simplifying and using that $|B_1|+|B_2|=|A|$, this gives
% Hence
% \begin{align*}
% f(x_i,x_j)
\\
&=\frac{\alpha}{\beta} \left(1+\frac{\kappa+x_i}{\beta} \right)^{-\alpha -1}
%\\
%&
+\frac{1}{\beta} \left(1+\frac{\kappa+x_i}{\beta}\right)^{b_{i \setminus j}-1} \left(1+\frac{2 \kappa+x_i}{\beta} \right)^{b_{i,j}-1} \left(1+\frac{\kappa}{\beta}\right)^{b_{i \setminus j}} \\
&\quad \times \left(-\alpha \left(1+\frac{\kappa+x_i}{\beta} \right) +\frac{b_{i \setminus j}\kappa}{\beta} \right).
\end{align*}
Fourth, we study the situation of two exceedances, so $X_i,X_j>0$: 
\begin{align*}
f(x_i,x_j)&=E_{\Lambda_i,\Lambda_j} \left[\Lambda_i e^{-(\kappa+x_i)\Lambda_i} \Lambda_j e^{-(\kappa+x_j) \Lambda_j} \right]
%\\
%&
= \frac{\partial}{\partial u_i} \frac{\partial}{\partial u_j}  \left. E_{\Lambda_i,\Lambda_j}  \left[e^{-u_i \Lambda_i -u_j \Lambda_j} \right] \right|_{\substack{u_i=\kappa+x_i,\\ u_j=\kappa+x_j}}\\
&=\frac{1}{\beta^2}  \left(1+\frac{2\kappa+x_i+x_j}{\beta} \right)^{b_{i,j}-2}
%\\
%&\times
\left(1+\frac{\kappa+x_i}{\beta} \right)^{b_{i\setminus j}-1}
 \left(1+\frac{\kappa+x_j}{\beta} \right)^{b_{i\setminus j}-1}\\
 &\quad  \times
\left[ 
b_{i,j}(b_{i,j}-1)  \left(1+\frac{\kappa+x_i}{\beta} \right)\left(1+\frac{\kappa+x_j}{\beta} \right) 
+b_{i \setminus j}^2\left(1+\frac{2\kappa+x_i+x_j}{\beta} \right)^2 
\right.\\
& \quad \left. 
+  b_{i,j}b_{i\setminus j}\left(1+\frac{2\kappa+x_i+x_j}{\beta} \right)
\left\{
\left(1+\frac{\kappa+x_i}{\beta} \right)
+ \left(1+\frac{\kappa+x_j}{\beta} \right) \right\}   \right].
\end{align*}

%%%%%%%%%%%%%%%%%%%%%%%%%%%%%%%%%%%%%%%%%%%%%%%%%%%%%%%%%%%%%%%%%%%%%%%%%%%%%%%%%%

\subsection{Proofs}\label{AProofs}
We present a fundamental result for integrals involving L\'evy bases, which we use repeatedly in our proofs. The result is derived in \citet[Proposition 2.6]{Rajput1989}, and stated here in our notation for a homogeneous L\'evy basis:
\begin{prop} \label{lbasisint}
Let $L$ be a homogeneous L\'evy basis and $f$ an $L$-integrable function, then
  \[E\left[\exp\left\{iu \int_S f(s) L(ds)\right\}\right]=\exp\left\{\int_S \cum{uf(s)}{L'}\right\} \, ds, \quad u \in \re.\]
\end{prop}
%\subsection{Proof of Proposition \ref{prop:jointchar}}

 \begin{proof}[Proof of Proposition \ref{prop:jointchar}]
We have 
$\sum_{j=1}^ku_j\Lambda_j= \sum_{j=1}^k u_j L(A_{t_j}) = \sum_{j=1}^k u_j \sum_{m: S_m\subset A_{t_j}} L(S_m) = \sum_{m=1}^{n_k} L(S_m)  u_{m}^+$.
Using the fact that the L\'{e}vy basis is independently scattered and Proposition \ref{lbasisint}, we get 
that $E\left(\exp\left(i\sum_{j=1}^k u_j \Lambda_j\right)\right) = E\left(\exp\left( i\sum_{m=1}^{n_k}  L(S_m)u_{m}^+ \right) \right)
= \prod_{m=1}^{n_k}  
E \left(\exp\left( i u_{m}^+ L(S_m) \right) \right)\\
= \prod_{m=1}^{n_k}    
\exp\left(\leb(S_m) C(u_{m}^+; L')\right)
=\exp\left(\sum_{m=1}^{n_k} \leb(S_k) C(u_{m}^+ ; L')\right)$. 
\end{proof}

%%%%%%%%%%%%%%%%%%%%%%%%%%%%%%%%%%%%%%%%%%%%%%%%%%%
%\subsection{Proof of Proposition \ref{acfexceedance}}
\begin{proof}[Proof of Proposition \ref{acfexceedance}]
The expected value of $X$ is obtained by conditioning on the events $\{X=0\}$ and $\{X>0\}$, with probabilities given by \eqref{exprob}, and using that $\{X|X>0\}$ has a GPD$(\alpha,\beta+\kappa)$ distribution, which gives
$E[X]=(1+\kappa/\beta)^{-\alpha}(\beta+\kappa)/(\alpha-1)$,  for   $\alpha>1$.

The expectation of the product is more complicated, as it involves the joint density of the corresponding values of the latent trawl process:
\begin{align*}
E[X_0X_h]=0&+\int_{\re_+^2} x_0 x_h \left( \int_{\re_+^2} f(x_0,x_h|\lambda_0,\lambda_h) f(\lambda_0,\lambda_h) \, d \lambda_0 d \lambda_h \right) \, dx_0 \, dx_h\\
&=\int_{\re_+^2} x_0 x_h \left( \int_{\re_+^2}  \lambda_0 e^{-\lambda_0(\kappa+x_0)} \lambda_h e^{-\lambda_h(\kappa+x_h)}  f(\lambda_0,\lambda_h) \, d \lambda_0 \, d \lambda_h \right) \, dx_0 \, dx_h.
\end{align*}
Using Fubini's theorem to repeatedly exchange the order of integration gives
\begin{align*}
E[X_0X_h]&= \int_{\re_+^2} \lambda_0 \lambda_h  e^{-(\lambda_0 \kappa+\lambda_h \kappa)}  f(\lambda_0,\lambda_h) \left(\int_{\re_+^2} x_0 x_h e^{-(x_0 \lambda_0 +x_h \lambda_h)} \, dx_0 \, dx_h \right) \, d \lambda_0 \, d \lambda_h\\
&= \int_{\re_+^2} \lambda_0 \lambda_h  e^{-(\lambda_0 \kappa+\lambda_h \kappa)}  f(\lambda_0,\lambda_h) \left(\frac{1}{\lambda_0^2\lambda_h^2}\right) \, d \lambda_0 \, d \lambda_h\\
&=E_{\Lambda_0,\Lambda_h} \left[\frac{e^{-(\Lambda_0 \kappa+\Lambda_h \kappa)}}{\Lambda_0 \Lambda_h}\right]= \int_\kappa^\infty \int_\kappa^\infty E_{\Lambda_0,\Lambda_h} \left[e^{-(\Lambda_0 u_0+\Lambda_h u_h)}\right]  \, du_0 \, du_h .
\end{align*}
The joint expectation can now be found using the partitioning method described above.  Specifically, partition the trawl sets $A_0,A_h$ into disjoint sets 
$B_{0\setminus h}=A_0 \setminus A_h,  B_{0,h}=A_0 \cap A_h,  B_{h\setminus 0}=A_h \setminus A_0$ 
and define the slice variables $\{S_i=L(B_i)\}$, where $L$ is the L\'evy basis derived from \eqref{gmargins}, so in particular these variables are independent. Furthermore, we have $\Lambda_0 \stackrel{d}{=} S_{0,h}+S_{0\setminus h}$ and $\Lambda_h \stackrel{d}{=}S_{0,h}+S_{h \setminus 0}$. 
Let $b_i=-\alpha \leb(B_i)/\leb(A)$ be the  parameter of $S_i$, then we can rewrite the above expression as 
\begin{align*}
E[X_0X_h]&= \int_\kappa^\infty \int_\kappa^\infty E_{S_{0\setminus h},S_{0,h},S_{h\setminus 0}}\left[\exp\{-\left[u_0 S_{0\setminus h}+(u_0+u_h) S_{0,h}+u_hS_{h\setminus 0} \right]\}\right]  \, du_0 \, du_h \\
&= \int_\kappa^\infty \int_\kappa^\infty \left(1+\frac{u_0}{\beta} \right)^{b_{0\setminus h}} \left(1+\frac{u_0+u_h}{\beta} \right)^{b_{0,h}} \left(1+\frac{u_h}{\beta} \right)^{b_{h\setminus 0}}  \, du_0 \, du_h. 
\end{align*}
\end{proof}
%%%%%%%%%%%%%%%%%%%%%%%%%%%%%%%%%%%%%%%%%%%%%%%%%%%%%%%%%%%%%%%%%%%%%%%%%%%%%%%%%%%%%%%%%%%%%%%%%

%%%%%%%%%%%%%%%%%%%%%%%%%%%%%
%\subsection{Proof of Proposition \ref{propcondtail}}
\begin{proof}[Proof of Proposition \ref{propcondtail}]
By conditioning on the latent trawl process we obtain
\begin{align*}
&P(F_{2e}(X_0)>u_1, F_{2e}(X_h)>u_2, X_0>0,X_h>0)\\
=&P(X_0>F^{-1}_{2e}(u_1) , X_h>F^{-1}_{2e}(u_2), X_0>0,X_h>0)\\
=&\int \exp\left\{-\lambda_0 (F^{-1}_{2e}(u_1)+\kappa)\right\} \exp\left\{-\lambda_h (F^{-1}_{2e}(u_2)+\kappa)\right\} f(\lambda_0,\lambda_h) \, d \lambda_0 d \lambda_h\\
=&E_{\Lambda_0,\Lambda_h} \exp\left\{-\Lambda_0 (F^{-1}_{2e}(u_1)+\kappa) -\Lambda_h (F^{-1}_{2e}(u_2)+\kappa)\right\}.
\end{align*}
This expectation can be calculated by using the partition representation of the trawl process, similar to the calculation in Section \ref{autocovariance}. Specifically, we use the partition
$B_{0\setminus h}=A_0 \setminus A_h, B_{0,h}=A_0 \cap A_h,  B_{h\setminus 0}=A_h \setminus A_0$,
with corresponding slice variables $S_i$ for $i\in\{(0\setminus h),(0,h),(h\setminus 0)\}$.
%, where $A_i$ is the trawl set at the observation time $t_i$ of the variable $X_i$. 
Letting $b_i$ denote the  parameter of the slice $S_i$ (as defined  in the proposition), we have
\begin{align}\begin{split}\label{helpful}
&P(F_{2e}(X_0)>u_1, F_{2e}(X_h)>u_2, X_0>0,X_h>0)\\
&=E_{S_{0\setminus h},S_{0,h},S_{h\setminus 0}} \left[  \exp\left\{-S_{0\setminus h} \left(F^{-1}_{2e}(u_1)+\kappa \right)-S_{0,h} \left(F^{-1}_{2e}(u_1)+F^{-1}_{2e}(u_2)+2\kappa \right)\right. 
\right.\\
&\left. 
\left. -S_{h\setminus 0}\left(F^{-1}_{2e}(u_2)+\kappa \right)\right\} \right]  \\
=&\left(1+\frac{\kappa+F^{-1}_{2e}(u_1)}{\beta} \right)^{b_{0\setminus h}} \left(1+\frac{2\kappa +F^{-1}_{2e}(u_1)+F^{-1}_{2e}(u_2)}{\beta} \right)^{b_{0,h}} \left(1+\frac{\kappa+F^{-1}_{2e}(u_2)}{\beta} \right)^{b_{h\setminus 0}}.
\end{split}     
\end{align}
Setting $u_2=0$ in this expression and noting that $F_{2e}^{-1}(0)=0$ shows that 
\begin{align} \label{3e}
\begin{split}
&P(F_{2e}(X_0)>u_1, X_0>0,X_h>0)\\
=&\left(1+\frac{\kappa+F^{-1}_{2e}(u_1)}{\beta} \right)^{b_{0\setminus h}} \left(1+\frac{2\kappa+F^{-1}_{2e}(u_1)}{\beta} \right)^{b_{0,h}} \left(1+\frac{\kappa}{\beta} \right)^{b_{h\setminus 0}},
\end{split}
\end{align}
giving the denominator of $\varphi$, so that we get
\[\varphi(h,u_1,u_2)=\left(1+\frac{F_{2e}^{-1}(u_2)}{\beta+2\kappa+F_{2e}^{-1}(u_1)}\right)^{b_{0,h}}\left(1+\frac{F_{2e}^{-1}(u_2)}{\beta+\kappa}\right)^{b_{h\setminus 0}}.\]  
We can also write
\[F_{2e}(x)=1-\frac{P(X_0>x,X_0>0,X_h>0)}{P(X_0>0,X_h>0)},\]
where the probabilities can be derived from \eqref{3e} by the substitution $u_1 \rightarrow F_{2e}(x)$, resulting in
\begin{align}\label{F2e}
F_{2e}(x)=1-\left(1+\frac{x}{\beta+2\kappa}\right)^{b_{0,h}}\left(1+\frac{x}{\beta+\kappa}\right)^{b_{0\setminus h}}.
\end{align}
\end{proof}
%%%%%%%%%%%%%%%%%%%%%%%%%%%%%%%%%%%%%%%%%%
\begin{lemma}\label{Fe2Lemma}
Consider the latent model as defined in Section \ref{setup}. Then the identity in \eqref{F2e_1} holds, i.e.~$P(X_0 \leq x\, | X_0>0, X_h>0)=P(X_h \leq x\, | X_0>0, X_h>0)$,  for  $h\in \mathbb{N}$.
\end{lemma}
\begin{proof}[Proof of Lemma \ref{Fe2Lemma}]
The proof of Proposition \ref{propcondtail} implies that (using the same notation as above) 
$F_{2e}(x)=P(X_0 \leq x\, | X_0>0, X_h>0) =1-\left(1+\frac{x}{\beta+2\kappa}\right)^{b_{0,h}}\left(1+\frac{x}{\beta+\kappa}\right)^{b_{0\setminus h}}$.
Similarly, we compute
$P(X_h \leq x\, | X_0>0, X_h>0) 
=1-P(X_h > x\, | X_0>0, X_h>0)$,
where  $P(X_h > x\, | X_0>0, X_h>0)=P(F_{2e}(X_h) > F_{2e}(x)\, | X_0>0, X_h>0)$. The latter expression can be derived from \eqref{helpful} by setting $u_2=F_{2e}(x)$ and $u_1=0$. Recall that   $F_{2e}^{-1}(0)=0$. Hence (using that $b_{0\setminus h}=b_{h\setminus 0}$ as we will show  in the proof of Proposition \ref{scalingofindex} below)
\begin{align} \label{3ee}
\begin{split}
&P(F_{2e}(X_h)>u_2, X_0>0,X_h>0)\\
&=\left(1+\frac{\kappa}{\beta} \right)^{b_{0\setminus h}} \left(1+\frac{2\kappa +F^{-1}_{2e}(u_2)}{\beta} \right)^{b_{0,h}} \left(1+\frac{\kappa+F^{-1}_{2e}(u_2)}{\beta} \right)^{b_{h\setminus 0}}\\
&=\left(1+\frac{\kappa}{\beta} \right)^{b_{0\setminus h}} \left(1+\frac{2\kappa +x}{\beta} \right)^{b_{0,h}} \left(1+\frac{\kappa+x}{\beta} \right)^{b_{0\setminus h}},
\end{split}
\end{align}
and
\begin{align*}
P(X_0>0,X_h>0)=\left(1+\frac{\kappa}{\beta} \right)^{2b_{0\setminus h}} \left(1+\frac{2\kappa }{\beta} \right)^{b_{0,h}}.
\end{align*}
Hence, we conclude that 
\begin{align*}
P(X_h \leq x\, | X_0>0, X_h>0)&=1-\left(1+\frac{x}{\beta+2\kappa}\right)^{b_{0,h}}\left(1+\frac{x}{\beta+\kappa}\right)^{b_{0\setminus h}}\\
&=P(X_0 \leq x\, | X_0>0, X_h>0).
\end{align*}
\end{proof}
%%%%%%%%%%%%%%%%%%%%%
%\subsection{Proof of Proposition \ref{scalingofindex}}
\begin{proof}[Proof of Proposition \ref{scalingofindex}]
Using the representation result from Proposition \ref{propcondtail}, we can immediately read off that  $\varphi(h,u_1,0)=1$. Moreover, from the definition of $F_{2e}$, we observe that 
$\varphi(h,0,u_2)=1-u_2$ for all $0\leq u_2\leq 1$ is equivalent to $\varphi(h,0,F_{2e}(x))=1-F_{2e}(x)$ for all $x\in \re$. The latter statement can be easily computed as follows. We now use representation \eqref{F2e} to deduce that, for any $x\in \re$, we have 
\begin{align*}
\varphi(h,0,F_{2e}(x))&= 
\left(1+\frac{x}{\beta+2\kappa}\right)^{b_{0,h}}\left(1+\frac{x}{\beta+\kappa}\right)^{b_{h\setminus 0}},\\
 1-F_{2e}(x) &=\left(1+\frac{x}{\beta+2\kappa}\right)^{b_{0,h}}\left(1+\frac{x}{\beta+\kappa}\right)^{b_{0\setminus h}}.
\end{align*}
Hence, we get the identity as soon as $b_{0\setminus h} = b_{h\setminus 0}$ which is equivalent to $\leb(A_0\setminus A_h) =\leb(A_h\setminus A_0)$. Since the trawl is constructed via translation, we have indeed that 
$\leb(A_0)=\leb(A_0\cap A_h)+\leb(A_0\setminus A_h)
= \leb(A_0\cap A_h)+\leb(A_h\setminus A_0)=\leb(A_h)$, which implies that $b_{0\setminus h} = b_{h\setminus 0}$.
\end{proof}
%%%%%%%%%%%%%%%%%%%%%%%%%%%%%%%%%%%%%%%%%
%\subsection{Proof of Theorem \ref{thm_phi}}
\begin{proof}[Proof of Theorem \ref{thm_phi}]
Note that
\begin{align}\label{varphi}
\varphi(h,u_1,u_2)&=\left(1+\frac{F_{2e}^{-1}(u_2)}{\beta+2\kappa+F_{2e}^{-1}(u_1)}\right)^{b_{0,h}}\left(1+\frac{F_{2e}^{-1}(u_2)}{\beta+\kappa}\right)^{b_{h\setminus 0}}.
% \\
% &=\left(1+\frac{1}{\frac{\beta+2\kappa}{F_{2e}^{-1}(u_2)}+\frac{F_{2e}^{-1}(u_1)}{F_{2e}^{-1}(u_2)}}\right)^{-b_{0,h}}\left(1+\frac{F_{2e}^{-1}(u_2)}{\beta+\kappa}\right)^{-b_h}.
\end{align}
Since $\lim_{u \uparrow 1}F_{2e}^{-1}(u)=\infty$, the first term in in \eqref{varphi} converges to $2^{b_{0,h}}$ and the  second term in \eqref{varphi} tends to 0 and hence we can deduce that $\lim_{u_1\uparrow 1, u_2\uparrow 2}\varphi(h,u_1,u_2)=0$.
% Also, if $b_h=0$, then $\lim_{u_1\uparrow 1, u_2\uparrow 2}\varphi(h,u_1,u_2)=2^{-b_{0,h}}$.
\end{proof}

%%%%%%%%%%%%%%%%%%%%%%%%%%%%%%%%%%%%%%%%%%%%%%%%%%%%%%%%%%%%%%%%%%%%%%%%%%%%%%%%%%%%%%

\end{appendix}
% \bibliography{PhD,books,extremes_books,MIDAS_data,air_quality_data,pedersen,additionalrefs}
% \bibliographystyle{agsm}

\end{document}